\documentclass[14pt, aps, preprint, nofootinbib]{revtex4}

\usepackage{graphicx, subfig, float, hyperref}
\usepackage{amsmath, amsfonts, amsthm, amssymb}
\newtheorem{fact}{Fact}

\begin{document}

\title{On the density of states in a free CFT and finite volume corrections}

\author{Connor Behan}
\affiliation{Department of Physics and Astronomy, University of British 
Columbia, Vancouver, BC, V6T 1W9, Canada}
\date{\today}

\begin{abstract}
Results from spectral geometry such as Weyl's formula can be used to relate the
thermodynamic properties of a free massless field to the spatial manifold on which
it is defined. We begin by calculating the free energy in two cases: manifolds
posessing a boundary and spheres. The subextensive contributions allow us to test
the Cardy-Verlinde formula and offer a new perspective on why it only holds in a
free theory if one allows for a change in the overall coefficient. After this
we derive an expression for the density of states that takes the form of a Taylor
series. This series leads to an improvement over known results when the
area of the manifold's boundary is nonzero but much less than the appropriate
power of its volume.
\end{abstract}

\newpage
\maketitle
\newpage

\section{Introduction}
The density of states for a quantum mechanical system is defined as a measure
$\rho(E)\textup{d}E$ giving the number of Hilbert space states with energy in
$[E, E + \textup{d}E)$. Strictly speaking it only exists for theories that have
a continuous spectrum, but as long as the energy times the length scale is large,
it becomes a convenient tool for counting states in a discrete spectrum as well.
In a free field theory, it is well known that having a continuous spectrum is
equivalent to the statement that the theory lives in a manifold $M$ that has an
infinite volume. The present work considers finite volume $|M|$ and is also
primarily concerned with spatial dimension $d > 1$.

Thermodynamic quantities in a $(1 + 1)$-dimensional conformal field theory are
well understood because the partition function in such a theory is invariant under
modular transformations. In 1986, Cardy used this to derive a formula for the
entropy \cite{cardy}
\begin{equation}
S = 2\pi \sqrt{\frac{c + \tilde{c}}{6} \left ( E - \frac{c + \tilde{c}}{24} 
\right )} ,\label{cardy-formula}
\end{equation}
where the volume is normalized to $2\pi$ by convention.
Letting the two central charges be equal and dropping terms with $c^2$ shows us
that the density of states has the following asymptotic expression:\footnote{
For an explanation of the notation, we will use $\propto$ in
equations that are roughly true to illustrate a conceptual point. For example
$f(x) \propto g(x)$ means that the leading term of $f$ is proportional to the
leading term of $g$. A symbol we use more rigorously is $\sim$. We say that
$f(x) \sim g(x)$ if $f(x) = g(x) + o(g(x))$. In other words, the leading terms
are equal and as we let $x \rightarrow \infty$,
$f - g$ becomes negligible compared to $f$ and $g$.}
\begin{equation}
\rho(E) \propto e^{2 \pi \sqrt{\frac{cE}{3}}} = e^{\sqrt{\frac{2c}{3} \pi |M| E}}
.\nonumber
\end{equation}
Efforts to extend the Cardy formula have been fruitful and have invoked powerful
theorems such as the Rademacher expansion for the Fourier coefficients of a
modular form \cite{farey}. In 2011, Loran, Sheikh-Jabbari and Vincon
proceeded along these lines to show that
\begin{eqnarray}
\rho(E) &\sim& \frac{2\pi^2 c}{3} 
\frac{I_1\left ( 2\pi\sqrt{cE / 3} \right )}{2 \pi \sqrt{cE / 3}} \nonumber \\
&=& \frac{\pi c |M|}{3} 
\frac{I_1\left ( \sqrt{2\pi c |M|E / 3} \right )}{\sqrt{2\pi c |M|E / 3}}
\label{bessel-function}
\end{eqnarray}
up to exponentially suppressed contributions \cite{loran}. Here $I_1$ is the
modified Bessel function of the first kind.

For the most part, powerful techniques based on modular invariance are no longer
available in higher dimensions. There are two main results for arbitrary dimension
and each makes different assumptions about the interactions in the theory.
One is the Cardy-Verlinde formula \cite{verlinde2} which has been used to describe
strongly coupled theories on a $d$-sphere with gravity duals
\cite{klemm, cai2, lin, brustein}. It states that
\begin{equation}
S = \frac{2\pi r}{d} \sqrt{E_{\textup{C}}(2E-E_{\textup{C}})} ,\nonumber
\end{equation}
where $E_{\textup{C}}$ is the Casimir energy. Since the Casimir energy of a CFT
compactified on a circle of radius $r$ is $\frac{c + \tilde{c}}{12r}$, this is in
perfect analogy with the Cardy formula
\begin{equation}
S = 2\pi r \sqrt{\frac{c + \tilde{c}}{12r}
\left ( 2E-\frac{c + \tilde{c}}{12r} \right )} ,\nonumber
\end{equation}
which is essentially (\ref{cardy-formula}). The other is the standard result for a
free theory where the partition function and free energy are given by
\begin{eqnarray}
Z &=& e^{\frac{A}{\beta^d}} \nonumber \\
F &=& -AT^{d + 1} \label{free-free-energy}
\end{eqnarray}
respectively. In a theory with $s$ species of bosons and $s^*$ species of
fermions, $A$ is given by
\begin{equation}
A = \frac{d! \omega_d [s\zeta(d + 1) + s^{*}\zeta^{*}(d + 1)] |M|}{(2\pi)^d}
,\nonumber
\end{equation}
where $\zeta(\sigma) = \sum_{n=1}^{\infty} \frac{1}{n^{\sigma}}$ is the Riemann
zeta function,
$\zeta^*(\sigma) = \sum_{n=1}^{\infty} \frac{(-1)^{n-1}}{n^{\sigma}}$ is the
alternating zeta function and $\omega_d$ is the volume of a unit ball in
$\mathbb{R}^d$.\footnote{Even though it is not usually stated in such a general
form, this is a widely published result. One example is the maximally
supersymmetric large $N$ Yang-Mills theory for which $d = 3$ and
$s = s^{*} = 8N^2$. In the free theory, $F = \frac{\pi^2 N^2 |M|}{6} T^4$
appears frequently in the literature \cite{gubser, kim, fotopoulos, cai1}. Another
example is QCD for which $d = 3$, $s = 16$ and $s^{*} = 12 n_{\textup{f}}$.
In the asymptotically free limit, \cite{braaten} has found that
$F = -\frac{8 \pi^2 |M|}{45} T^4 \left (1 + \frac{21}{32} n_{\textup{f}} \right )$
which also agrees with (\ref{free-free-energy}).}

An expression for $\rho(E)$ can be calculated from (\ref{free-free-energy})
because the density of states is the inverse Laplace transform of the partition
function:
\begin{equation}
\rho(E) = \frac{1}{2\pi} \int_{-\infty}^{\infty} Z(i\beta) e^{i\beta E}
\textup{d}\beta .\nonumber
\end{equation}
Integrals of the form $\int_{-\infty}^{\infty} e^{f(i\beta)} \textup{d}\beta$ can
be approximated as $e^{f(i\beta_0)}$ where $i\beta_0$ is a saddle point of $f$.
However, if we want our expression to have an inverse energy prefactor as is
required for a proper density of states, we must consider quadratic fluctuations
around the saddle point and evaluate a Gaussian integral. The density of states
becomes
\begin{eqnarray}
\rho(E) &\sim& \frac{1}{2\pi} \int_{-\infty}^{\infty}
e^{f(i\beta_0) - \frac{\beta^2}{2} f^{\prime \prime}(i\beta_0)} \textup{d}\beta
\nonumber \\
&=& \frac{1}{\sqrt{2\pi f^{\prime \prime}(i\beta_0)}} e^{f(i\beta_0)} .\nonumber
\end{eqnarray}
Since $f(i\beta) = i\beta E + \frac{A}{(i\beta)^d}$, the saddle point can be found
as $i\beta_0 = \left ( \frac{dA}{E} \right )^{\frac{1}{d+1}}$ giving
\begin{equation}
\rho(E) \sim \frac{1}{\sqrt{2\pi}} \sqrt{\frac{d^{\frac{1}{d+1}}}{d+1}}
A^{\frac{1}{2(d+1)}} E^{\frac{-d-2}{2(d+1)}}
\exp \left( \frac{d+1}{d^{\frac{d}{d+1}}} A^{\frac{1}{d+1}} E^{\frac{d}{d+1}}
\right) \label{2terms-2approximations}
\end{equation}
for the density of states.\footnote{The fact that $\rho(E)$ must grow
exponentially as $E^{\frac{d}{d+1}}$ is already clear from dimensional
analysis. Energy and entropy are both extensive so they must be proportional
to the volume. This means that $E \propto |M|T^{d+1}$ and $S \propto |M|T^d$
because in a CFT, the only other scale is the one set by temperature. Substituting
these into eachother gives $S \propto V^{\frac{1}{d+1}} E^{\frac{d}{d+1}}$.}

Two approximations are being made in this established scheme which we will call
the saddle point approximation and the continuum approximation. The saddle point
approximation is the truncation of the exponent in the integral and it is useful
because the integral has no closed form expression. The continuum approximation
is the assumption that we may replace our sums with integrals and use the
partition function for a theory that has a continuous spectrum. The goal of this
work is to improve upon each of these approximations separately. In the first
section, we analyze the spectrum of the Laplacian for three different field
configurations. The family of results we derive goes beyond the continuum
approximation by providing corrections to the fact that the energy times
the length scale is large. They are:
\begin{eqnarray}
Z &=& e^{\frac{A}{\beta^d} + \frac{B}{\beta^{d-1}}} \nonumber \\
F &=& -AT^{d + 1} - BT^d \nonumber \\
B &=& \frac{(d-1)! \omega_{d-1} [s\zeta(d) + s^{*}\zeta^{*}(d)] |\partial M|}
{4(2\pi)^{d-1}} ,\label{main-result1}
\end{eqnarray}
for minimally coupled fields on a general manifold with boundary $\partial M$,
\begin{eqnarray}
Z &=& e^{\frac{A}{\beta^d} + \frac{C}{\beta^{d-2}}} \nonumber \\
F &=& -AT^{d + 1} - CT^{d-1} \nonumber \\
C &=& \frac{d}{6} [s\zeta(d-1) + s^{*}\zeta^{*}(d-1)] r^{d-2}
,\label{main-result2}
\end{eqnarray}
for minimally coupled fields on a sphere and
\begin{eqnarray}
Z &=& e^{\frac{A}{\beta^d} + \frac{C^{\prime}}{\beta^{d-2}}} \nonumber \\
F &=& -AT^{d + 1} - C^{\prime}T^{d-1} \nonumber \\
C^{\prime} &=& \frac{-(d-3)^2}{4} [s\zeta(d-1) + s^{*}\zeta^{*}(d-1)] r^{d-2}
,\label{main-result3}
\end{eqnarray}
for conformally coupled fields on a sphere. In the second section, we show how
the finite size effects in our results provide an alternative route to known
results about the Cardy-Verlinde formula. The third section returns to the task
of going beyond the continuum approximation and discusses the difficulty in
converting the corrections for $Z(\beta)$ in
(\ref{main-result1}, \ref{main-result2}, \ref{main-result3}) into corrections for
$\rho(E)$. For the last section, we turn to the saddle point approximation. We go
beyond it by computing the inverse Laplace transform of $e^{A/\beta^d}$ as an
exact Taylor series. From this we arrive at
\begin{equation}
\rho(E) \sim \frac{1}{\sqrt{2\pi}} \sqrt{\frac{d^{\frac{1}{d+1}}}{d+1}}
A^{\frac{1}{2(d+1)}} E^{\frac{-d-2}{2(d+1)}}
\exp \left( \frac{d+1}{d^{\frac{d}{d+1}}} A^{\frac{1}{d+1}} E^{\frac{d}{d+1}}
- \frac{(d+2)(2d+1)}{24(d+1)} (dA)^{-\frac{1}{d+1}}
E^{-\frac{d}{d+1}} \right) ,\label{main-result4}
\end{equation}
extending (\ref{2terms-2approximations}). Although it is assumed that $B = 0$
when deriving this, it is possible for the correction in the exponent
to still be significant on a manifold with a boundary. Conversely, we will show
that it cannot be trusted on a sphere; the assumption $C = 0$ necessarily sets
the correction in the exponent to zero as well.

\section{Beyond the continuum approximation I}
A field theory can be regarded as a collection of degrees of freedom
labelled by a momentum $\textbf{p}$ and possibly other quantum numbers
labelling the spin and/or charge. When a particular degree of freedom
having squared momentum $p^2$ is excited, the energy added to the system is
$\sqrt{p^2 + m^2}$ by the relativistic dispersion relation. In a bosonic
theory, this excitation can be repeated infinitely many times, so if we
restrict our attention to this degree
of freedom, the single mode partition function is
\begin{eqnarray}
Z(p^2) &=& 1 + e^{-\beta \sqrt{p^2 + m^2}} + e^{-2\beta \sqrt{p^2 + m^2}} + 
\dots \nonumber \\
&=& \left ( 1 - e^{-\beta \sqrt{p^2 + m^2}} \right )^{-1} 
.\label{bosonic-first-step}
\end{eqnarray}
How many degrees of freedom have a squared momentum of $p^2$?
This question depends on the manifold $M$. Allowed momenta are eigenvalues of
the Laplace-Beltrami operator. This is a map between the two function spaces
$\Delta: C^2_0(M) \rightarrow C^0(M)$ where the
subscript $0$ indicates Dirichlet boundary conditions.
The eigenvalue equation is
\begin{equation}
\Delta f = -\lambda f \nonumber
\end{equation}
where $\lambda = p^2$. Understanding the spectrum of this operator is related
to Kac's problem; ``Can you hear the shape of a drum?'' Weyl's formula - one of
the earliest positive answers to this question - will be immensely useful
to us. The formula found by Weyl states that $\#(\lambda)$, the number of
Dirichlet eigenvalues up to $\lambda$ (for large $\lambda$) is given by:
\begin{equation}
\#(\lambda) = (2\pi)^{-d} \omega_d |M| \lambda^{\frac{d}{2}} + 
O \left( \lambda^{\frac{d - 1}{2}} \right ) ,\label{weyl}
\end{equation}
where $|M|$ is the volume of the manifold and $\omega_d$ is the volume of a unit
ball in $\mathbb{R}^d$ \cite{ivrii}. By differentiating this quantity, we can
determine $g(\lambda)$ - the number of eigenvalues between $\lambda$ and
$\lambda + \textup{d}\lambda$.
While $g(\lambda)$ gives us a degeneracy\footnote{Although $g(\lambda)$ gives the
degeneracy of the eigenvalue $\lambda$ and $\lambda$ is a way of labelling energy,
$g(\lambda)$ is \textit{not} the same as $\rho(E)$. In single particle problems,
Weyl's formula is sometimes referred to as a formula for the density of states
\cite{gutzwiller}, but in this problem, we are interested in the density of states
for the entire Fock space. We will therefore use the term \textit{degeneracy} for
$g(\lambda)$ and reserve the term \textit{density of states} for $\rho(E)$.} for
the eigenvalue $\lambda$ coming from geometrical considerations, we know that a
field theory can introduce further degeneracy as part of its description.
Therefore, if $s$ is the number of internal states, there are really
$s g(\lambda)$ ways of exciting a field to an eigenstate that has a squared
momentum of $\lambda$. These excitations contribute a factor of
\begin{equation}
Z(\lambda)^{s g(\lambda)} = 
\left ( 1 - e^{-\beta \sqrt{\lambda + m^2}} \right )^{-s g(\lambda)} \nonumber
\end{equation}
to the partition function. The full partition function we seek is the product
of $Z(\lambda)$ over all eigenvalues $\lambda$. Taking the log of a product
turns it into a sum, so
$\log{Z} = \sum_{\lambda} s g(\lambda) \log{Z(\lambda)}$. In the large volume
limit, the spectrum of the Laplacian becomes continuous, meaning that our sum
over $\lambda$ becomes an integral:
\begin{equation}
\log{Z} = \int_0^{\infty} s g(\lambda) \log{Z(\lambda)} \textup{d}\lambda 
.\nonumber
\end{equation}
If we proceeded to use (\ref{weyl}) in this calculation, we would derive the
common result (\ref{free-free-energy}). With this $g(\lambda)$, it is clear that
our answer would depend on the manifold only through its volume. This is to be
expected because at high enough temperatures, a field theory is insensitive to the
details of topology. For this reason, the calculation of (\ref{free-free-energy})
is usually done in flat space. Weyl's formula (\ref{weyl}) is simply a way
to justify this convenient choice. However, there is a correction to Weyl's
formula found by Ivrii in 1980 \cite{ivrii}:
\begin{equation}
\#(\lambda) = (2\pi)^{-d} \omega_d |M| \lambda^{\frac{d}{2}} + 
\frac{1}{4} (2\pi)^{-(d-1)} \omega_{d-1} |\partial M| \lambda^{\frac{d - 1}{2}}
+ o\left( \lambda^{\frac{d-1}{2}} \right ) .\label{weyl-ivrii}
\end{equation}
The second term, proportional to the boundary area, becomes important when the
length scale is not infinitely large and therefore contains information about
the lower temperature thermodynamics as well. Using this term,
\begin{eqnarray}
g(\lambda) &=& \frac{\textup{d}\#}{\textup{d}\lambda} \nonumber \\
&\sim& \frac{d}{2} (2\pi)^{-d} \omega_d |M| \lambda^{\frac{d-2}{2}}
+ \frac{d-1}{8} (2\pi)^{-(d-1)} \omega_{d-1} |\partial M| \lambda^{\frac{d-3}{2}}
.\nonumber
\end{eqnarray}
Now we can continue calculating the partition function.
\begin{eqnarray}
\log{Z} &=& \int_0^{\infty} s g(\lambda) \log{Z(\lambda)} \textup{d}\lambda 
\nonumber \\
&=& -\int_0^{\infty} \left ( \frac{sd \omega_d |M|}{2 (2\pi)^d}
\lambda^{\frac{d-2}{2}} + \frac{s(d-1) \omega_{d-1} |\partial M|}{8 (2\pi)^{d-1}}
\lambda^{\frac{d-3}{2}} \right )
\log\left ( 1 - e^{-\beta \sqrt{\lambda + m^2}} \right ) 
\textup{d}\lambda \nonumber \\
&=& \sum_{n = 1}^{\infty} 
\int_0^{\infty} \left ( \frac{sd \omega_d |M|}{2 (2\pi)^d}
\lambda^{\frac{d-2}{2}} + \frac{s(d-1) \omega_{d-1} |\partial M|}{8 (2\pi)^{d-1}}
\lambda^{\frac{d-3}{2}} \right ) \frac{1}{n} e^{-n \beta \sqrt{\lambda + m^2}} 
\textup{d}\lambda \nonumber \\
&=& \sum_{n = 1}^{\infty} 
\int_0^{\infty} \left ( \frac{sd \omega_d |M|}{(2\pi)^d}
p^{d - 1} + \frac{s(d-1) \omega_{d-1} |\partial M|}{4 (2\pi)^{d-1}}
p^{d - 2} \right ) \frac{1}{n} e^{-n \beta \sqrt{p^2 + m^2}} 
\textup{d}p \nonumber
\end{eqnarray}
In the second last step, we have Taylor expanded the logarithm and in the last
step, we have used the fact that $\lambda = p^2$. To proceed further, we must
make the field theory (globally) conformal by setting the mass to zero. This
allows us to use the identity
$\int_0^{\infty} e^{-cx} x^{d - 1} \textup{d}x = \frac{(d - 1)!}{c^d}$.
\begin{eqnarray}
\log{Z} &=& \frac{sd! \omega_d |M|}{(2\pi\beta)^d}
\sum_{n = 1}^{\infty} \frac{1}{n^{d + 1}} +
\frac{s(d-1)! \omega_{d-1} |\partial M|}{4 (2\pi\beta)^{d-1}}
\sum_{n = 1}^{\infty} \frac{1}{n^{d}} \nonumber \\
&=& \frac{sd! \omega_d |M| \zeta(d+1)}{(2\pi\beta)^d}
+ \frac{s(d-1)! \omega_{d-1} |\partial M| \zeta(d)}{4 (2\pi\beta)^{d-1}}
\label{bosonic-last-step}
\end{eqnarray}
The analysis so far has been applied to a bosonic theory, but very little
changes when applying it to fermions. Each mode can have 0 or 1 excitations
so instead of (\ref{bosonic-first-step}), we have:
\begin{equation}
Z(\lambda) = 1 + e^{-\beta \sqrt{\lambda}} .\label{fermionic-first-step}
\end{equation}
Proceeding to calculate (\ref{bosonic-last-step}) in the same way, we get:
\begin{eqnarray}
\log{Z} &=& \int_0^{\infty} s^{*} g(\lambda) \log{Z(\lambda)} 
\textup{d}\lambda \nonumber \\
&=& \int_0^{\infty} \left ( \frac{s^{*}d \omega_d |M|}{2 (2\pi)^d}
\lambda^{\frac{d-2}{2}}+\frac{s^{*}(d-1) \omega_{d-1}|\partial M|}{8 (2\pi)^{d-1}}
\lambda^{\frac{d-3}{2}} \right ) \log\left( 1 + e^{-\beta \sqrt{\lambda}} \right)
\textup{d}\lambda \nonumber \\
&=& \sum_{n = 1}^{\infty} \int_0^{\infty}
\left ( \frac{s^{*}d \omega_d |M|}{(2\pi)^d} p^{d-1} +
\frac{s^{*}(d-1) \omega_{d-1} |\partial M|}{4 (2\pi)^{d-1}}
p^{d-2} \right ) \frac{(-1)^{n + 1}}{n} e^{-n \beta p} 
\textup{d}p \nonumber \\
&=& \frac{s^{*}d! \omega_d |M| \zeta^{*}(d+1)}{(2\pi\beta)^d}
+ \frac{s^{*}(d-1)! \omega_{d-1} |\partial M| \zeta^{*}(d)}{4 (2\pi\beta)^{d-1}}
,\label{fermionic-last-step}
\end{eqnarray}
where $\zeta^{*}(\sigma) = \left ( 1 - 2^{1 - \sigma} \right ) \zeta(\sigma)$
is the alternating zeta function.

By adding (\ref{bosonic-last-step}) and (\ref{fermionic-last-step}) together, we
have shown the first of our results (\ref{main-result1}) - a free energy with an
extensive term and a subextensive term. The extensive free energy, which comes
from Weyl's original formula, is that which would be derived for a continuous
spectrum. The subextensive term, which comes from Ivrii's correction, extends the
validity to lower energies where one might start to notice the discrete spectrum.
The fact that this term is proportional to the area
of the manifold's boundary, introduces a problem in the common case of a
boundaryless manifold. If $|\partial M| = 0$, any subextensive contributions would
have to come from a third term in Weyl's formula. Such a term is not known so
instead of repeating our calculation of the partition function for a general
boundaryless manifold we will specialize to the case of a $d$-sphere where the
Laplace eigenvalue problem has been solved.
The $d$-dimensional spherical harmonics obey the eigenvalue equation
\begin{equation}
r^2 \Delta Y_{l_1, \dots, l_d} = -l_d (l_d + d - 1) Y_{l_1, \dots, l_d} 
\label{spherical-harmonics}
\end{equation}
and have integer indices that satisfy
$0 \leq |l_1| \leq l_2 \leq \dots \leq l_d$ \cite{higuchi}. If the condition
were $l_1 \leq \dots \leq l_d$, we could use the formula for $l_d$-simplex
numbers to solve for the degeneracy of each
eigenvalue as $\binom{l_d + d - 1}{d - 1}$. However, $l_1$ is allowed to be
negative. Multiplying by $2$ and subtracting
the number of $l_2 \leq \dots \leq l_d$ choices so as not to double count
$l_1 = 0$, we find a degeneracy given by:
\begin{equation}
g(l_d) = 2\binom{l_d + d - 1}{d - 1} - \binom{l_d + d - 2}{d - 2} = 
\frac{2l_d + d - 1}{d - 1} \binom{l_d + d - 2}{d - 2} .\label{sphere-g}
\end{equation}
Using the eigenvalue equation and the degeneracy, we may write:
\begin{eqnarray}
\frac{2l_d + d - 1}{d - 1} \binom{l_d + d - 2}{d - 2} \textup{d}l_d &=& 
\frac{\sqrt{(d-1)^2 + 4r^2p^2}}{d - 1} \binom{l_d + d - 2}{d - 2} \textup{d}l_d
\nonumber \\
&=& \frac{\sqrt{(d-1)^2 + 4r^2p^2}}{d - 1} \binom{l_d + d - 2}{d - 2} 
\frac{2r^2p}{\sqrt{(d-1)^2 + 4r^2p^2}} \textup{d}p \nonumber \\
&=& \frac{2r^2p}{d - 1} \binom{l_d + d - 2}{d - 2} \textup{d}p .\nonumber
\end{eqnarray}
Carrying out the same procedure as before, we have:
\begin{eqnarray}
\log Z &=& -\frac{2s}{d - 1} \int_0^{\infty} r^2 p \binom{l_d + d - 2}{d - 2} 
\log \left ( 1 - e^{-\beta p} \right ) \textup{d}p \nonumber \\
&=& -\frac{2s}{(d - 1)!} \int_0^{\infty} r^2 (l_d + d - 2) (l_d + d - 3) 
\dots (l_d + 1) \log \left ( 1 - e^{-\beta p} \right ) p \textup{d}p 
\nonumber \\
&=& -\frac{2s}{(d - 1)!} \int_0^{\infty} \frac{r^2}{2^{d - 2}} 
(\sqrt{4r^2p^2 + (d - 1)^2} + d - 3) (\sqrt{4r^2p^2 + (d - 1)^2} + d - 5) 
\nonumber \\
&& \dots (\sqrt{4r^2p^2 + (d - 1)^2} + 3 - d) 
\log \left ( 1 - e^{-\beta p} \right ) p \textup{d}p .\nonumber
\end{eqnarray}
Notice that the product in the integrand is already factored as a difference 
of squares. If $d = 2$, it is the empty product. If $d$ is an even number
greater than 2, it is a product of $(4r^2p^2 + (d-1)^2 - (d-3)^2)$,
$(4r^2p^2 + (d-1)^2 - (d-5)^2)$, \textit{etc}. If $d$ is odd, there is an
extra $\sqrt{4r^2p^2 + (d-1)^2}$ left over which we Taylor expand as
$2rp + \frac{(d-1)^2}{4rp}$. Collecting all the highest powers of $r$ in
the integrand, we get $r^d p^{d - 1}$ - the term proportional to the volume.
We would now like to find the term with the next highest power of $r$.
In the even case it is
\begin{equation}
\frac{1}{4} r^{d-2} p^{d-3} \left [ \left ( \frac{d}{2} - 1 \right ) (d - 1)^2 
- \sum_{k = 0}^{\frac{d-4}{2}} (2k + 1)^2 \right ] = \frac{1}{12} 
r^{d-2} p^{d-3} ( d^3 - 3d^2 + 2d ) \nonumber
\end{equation}
and in the odd case, it is
\begin{equation}
\frac{1}{4} r^{d-2} p^{d-3} \left [ \left ( \frac{d}{2} - 1 \right ) 
(d - 1)^2 - \sum_{k = 0}^{\frac{d-3}{2}} (2k)^2 \right ] = 
\frac{1}{12} r^{d-2} p^{d-3} ( d^3 - 3d^2 + 2d ) .\nonumber
\end{equation}
Unsurprisingly, for even and odd dimension, the same polynomial in $d$ appears
in the degeneracy. Moreover, it factors as $d(d-1)(d-2)$. We can therefore
evaluate
\begin{eqnarray}
\log Z &=& \frac{-2s}{(d-1)!} \int_0^{\infty}
\left ( r^d p^{d-1} + \frac{1}{12}d(d-1)(d-2)r^{d-2}p^{d-3} \right )
\log \left ( 1 - e^{-\beta p} \right ) \textup{d}p \nonumber \\
&=& \frac{2 r^d s\zeta(d+1)}{\beta^d} + 
\frac{d r^{d-2} s\zeta(d-1)}{6\beta^{d-2}} .\nonumber
\end{eqnarray}
If we were to derive this again for fermions, $s$ would become $s^{*}$
and $\zeta$ would become $\zeta^{*}$. We have therefore derived the second of
our results (\ref{main-result2}). This again has a free energy with an
extensive (contiunuum approximation) term and a subextensive (beyond the
continuum approximation) term, but now the subextensive term is proportional
to a power of the sphere's radius, which does not vanish.

\section{The Cardy-Verlinde formula}
A generalization of the Cardy formula was proposed in 2000 by Erik Verlinde
\cite{verlinde2}. The main work investigating its use in a free theory is a 2001
paper by Kutasov and Larsen \cite{kutasov}. They computed partition functions
for a number of free theories on $\mathbb{S}^3$ and $\mathbb{S}^5$ and showed that
Cardy-Verlinde does not hold for them. Specifically they demonstrated that it
predicts the correct scaling, but with a different coefficient.
In addition to commenting on their 3 and 5-dimensional results, we will write the
analogous statement for a general $d$. Before specifying
coefficients, Verlinde introduced the formula as:
\begin{equation}
S = \frac{2\pi r}{\sqrt{ab}} \sqrt{E_{\textup{C}}(2E - E_{\textup{C}})}
,\label{original-cardy-verlinde}
\end{equation}
where $a$ and $b$ are dimensionless constants. This can be seen if we split the
energy into extensive $E_{\textup{E}}$ and subextensive $E_{\textup{C}}$ parts
and apply the following relations:
\begin{eqnarray}
E &=& E_{\textup{E}} + \frac{1}{2}E_{\textup{C}} \nonumber \\
E_{\textup{E}} &=& \frac{a}{4\pi r} S^{1+\frac{1}{d}} \nonumber \\
E_{\textup{C}} &=& \frac{b}{2\pi r} S^{1-\frac{1}{d}} .\label{casimir-relations}
\end{eqnarray}
The power law relating the entropy to the extensive energy is clear.
$S \propto T^d$ and $E_{\textup{E}} \propto T^{d+1}$. Therefore the interesting
part of (\ref{casimir-relations}) is the statement the $E_{\textup{C}}$ should be
proportional to $T^{d-1}$.\footnote{In any number of dimensions, the vacuum energy
between two parallel plates is nonzero at zero temperature \cite{brevik}.
Verlinde's definition of $E_{\textup{C}}$ however only has this property when
$d = 1$.} Since the Casimir energy is proportional to the
subextensive free energy, our results (\ref{main-result1}, \ref{main-result2})
show that (\ref{casimir-relations}) is not
satisfied on a manifold that has a boundary. We can interpret the results of
Kutasov and Larsen as showing the $E_{\textup{C}} \propto S^{1-\frac{1}{d}}$
scaling precisely because they used a boundaryless manifold for their geometry.

To comment on their results further, we should calculate $E_{\textup{C}}$ on a
sphere for arbitrary values of $d$. The Casimir energy can be written
\begin{equation}
E_{\textup{C}} = dF + E = (d+1)F + TS \nonumber
\end{equation}
as it is the deviation of the energy from the Euler identity. If we tried this
with our result (\ref{main-result2}) for minimally coupled fields, we would find
$E_{\textup{C}} = -2CT^{d-1}$. This is a problem because it is negative for $d>1$
and (\ref{original-cardy-verlinde}) makes no sense for a negative Casimir energy.
The reason why the result of \cite{kutasov} is non-trivial is because the Casimir
energy is positive for \textit{conformally} coupled fields. The equation of
motion for a scalar field coupled to the background curvature $R$ is
$\Delta \phi + \xi R \phi = 0$. Spheres have constant curvature so this does not
make the analysis harder. Using $R = \frac{d(d-1)}{r^2}$ for spheres and
$\xi = \frac{d-1}{4d}$ for conformal coupling, squared momenta are no longer
eigenvalues of $-\Delta$ but of
\begin{equation}
-\Delta + \frac{(d-1)^2}{4r^2} .\nonumber
\end{equation}
Also,
\begin{equation}
p^2 = \frac{l_d (l_d + d - 2)}{r^2} + \frac{(d-2)^2}{4r^2} =
\frac{\left ( l_d + \frac{d-1}{2} \right )^2}{r^2} \nonumber
\end{equation}
from (\ref{spherical-harmonics}). The eigenvalue degeneracy $g(l_d)$ is still
(\ref{sphere-g}). After making the substitution $l = l_d + \frac{d-1}{2}$,
our partition function is given by
\begin{eqnarray}
\log Z &=& \frac{2s}{d-1} \int_0^{\infty} l \binom{l + \frac{d-3}{2}}{d-2}
\log \left ( 1 - e^{-\beta l / r} \right ) \textup{d}l \nonumber \\
&=& \frac{2s}{d-1} \int_0^{\infty} l \left ( l + \frac{d-3}{2} \right )
\left ( l + \frac{d-5}{2} \right ) \dots \left ( l - \frac{d-3}{2} \right )
\log \left ( 1 - e^{-\beta l / r} \right ) \textup{d}l \nonumber \\
&=& \frac{2s}{d-1} \int_0^{\infty} l \left (l^2 - \frac{(d-3)^2}{4} \right )
\left (l^2 - \frac{(d-5)^2}{4} \right ) \dots
\log \left ( 1 - e^{-\beta l / r} \right ) \textup{d}l \nonumber \\
&=& \frac{2s}{d-1} \int_0^{\infty} \left [ l^{d-1}
- \sum_{k=1}^{\frac{d-1}{2}} \left ( \frac{d-1}{2} - k \right )
l^{d-3} + O\left( l^{d-5} \right ) \right ]
\log \left ( 1 - e^{-\beta l / r} \right ) \textup{d}l \nonumber \\
&=& \frac{2s}{d-1} \int_0^{\infty} \left [ l^{d-1}
- \frac{1}{8} (d - 3) (d - 1) l^{d-3} + O\left( l^{d-5} \right ) \right ]
\log \left ( 1 - e^{-\beta l / r} \right ) \textup{d}l \nonumber \\
&\approx& \frac{2sr^d}{\beta^d} \zeta(d+1) -
\frac{s(d-3)^2 r^{d-2}}{\beta^{d-2}} \zeta(d-1) \nonumber
\end{eqnarray}
Doing this again for fermions gives (\ref{main-result3}) and it is not hard
to calculate the extensive and subextensive energies from this:
\begin{eqnarray}
E_{\textup{E}} = 2dA T^{d+1} \nonumber \\
E_{\textup{C}} = -2C^{\prime}T^{d-1} .\nonumber
\end{eqnarray}
When $d \neq 3$, $C^{\prime}$ is negative. Comparing this to the entropy, which is
\begin{equation}
S = (d + 1) A T^d \nonumber
\end{equation}
to leading order, we see that the $a$ and $b$ constants in
(\ref{casimir-relations}) are given by:
\begin{eqnarray}
\frac{a}{4\pi r} &=& \frac{2A}{\left ((d+1)A\right )^{1+\frac{1}{d}}} \nonumber \\
\frac{b}{2\pi r} &=& \frac{-2C^{\prime}}
{\left ((d+1)A\right )^{1-\frac{1}{d}}} .\label{a-and-b}
\end{eqnarray}
We have reproduced the result of \cite{kutasov} where in $d = 3$,
$E_{\textup{C}} \propto 1$ and the general form of (\ref{original-cardy-verlinde})
fails to hold. In all other dimensions (\ref{original-cardy-verlinde}) holds
with $a$ and $b$ given as in (\ref{a-and-b}). This is not what is usually called
the Cardy-Verlinde formula because the holographic theories studied by Verlinde
satisfy $\sqrt{ab} = d$ \cite{verlinde2}. The overall coefficient in free
theories however is
\begin{eqnarray}
\sqrt{ab} &=& \frac{4 \sqrt{2} \pi r}{d + 1} \sqrt{-\frac{C^{\prime}}{A}}
\nonumber \\
&=& 2 \pi \frac{d - 3}{d + 1}
\sqrt{\frac{s \zeta(d-1) + s^{*} \zeta^{*}(d-1)}
{s \zeta(d+1) + s^{*} \zeta^{*}(d+1)}} \nonumber
\end{eqnarray}
which, as found by Kutasov and Larsen, depends on the matter content.

\section{Beyond the continuum approximation II}
Earlier we wrote the density of states as the inverse Laplace transform of
$e^{f(i\beta)}$. We now wish to incorporate the subextensive contributions to the
partition function that we have found into the density of states.
Our expression for $f(i\beta)$ changes to
$f(i\beta) = i\beta E + \frac{A}{(i\beta)^d} + \frac{B}{(i\beta)^{d-1}}$ or
$f(i\beta) = i\beta E + \frac{A}{(i\beta)^d} + \frac{C}{(i\beta)^{d-2}}$
depending on whether we are describing a manifold with a boundary or a
sphere.\footnote{Or a sphere with conformal coupling where $C$ changes to
$C^{\prime}$.}
Previously we had $i\beta_0 = \left ( \frac{dA}{E} \right )^{\frac{1}{d+1}}$
as our saddle point but this is no longer true when discussing subextensive
corrections. Solving $f^{\prime}(i\beta_0) = 0$ algebraically for a general $d$
is not possible in either of the above cases. Instead, one can use Newton's
method to get a sense of how the saddle point shifts when going beyond the
continuum approximation. We know that
\begin{equation}
i\beta_{0, 0} = \left ( \frac{dA}{E} \right )^{\frac{1}{d+1}} \nonumber
\end{equation}
is the value of $i\beta_0$ when finite volume corrections are negligible. If these
contributions are not negligible but small, the corrected value of $i\beta_0$
should be close to this. Using $i\beta_{0, 0}$ as a starting point,
\begin{equation}
i\beta_{0, n+1} = i\beta_{0, n} - \frac{f^{\prime}(i\beta_{0, n})}
{f^{\prime\prime}(i\beta_{0, n})} \nonumber
\end{equation}
iterates Newton's method. The result of doing this once is
\begin{equation}
i\beta_0 = \left ( \frac{dA}{E} \right )^{\frac{1}{d+1}} \left [ 1 +
\left ( d + \frac{d+1}{d-1} \frac{(dA)^{\frac{d}{d+1}}}{B} E^{\frac{1}{d+1}}
\right )^{-1} \right ] ,\nonumber
\end{equation}
when the manifold has a boundary and
\begin{equation}
i\beta_0 = \left ( \frac{dA}{E} \right )^{\frac{1}{d+1}} \left [ 1 +
\left ( d-1 + \frac{d+1}{d-2} \frac{(dA)^{\frac{d-1}{d+1}}}{C} E^{\frac{2}{d+1}}
\right )^{-1} \right ] ,\nonumber
\end{equation}
when it is a sphere. These can be inserted into the saddle point formula
\begin{equation}
\rho(E) \sim \frac{1}{\sqrt{2\pi f^{\prime \prime}(i\beta_0)}} e^{f(i\beta_0)}
\nonumber
\end{equation}
which we will not do since it leads to a rather long and
unimaginative expression. We note that the exact saddle point can be found
without much work for minimally coupled fields on $\mathbb{S}^3$. The equation
that must be solved is
\begin{equation}
E\beta_0^4 + C\beta_0^2 -3A = 0 ,\nonumber
\end{equation}
which is a quartic over a quadratic, leading to
\begin{eqnarray}
\rho(E) \sim \frac{1}{\sqrt{2\pi}} \left( \frac{C+\sqrt{C^2+12AE}}{2E} \right)^{3/4}
\sqrt{\frac{24AE}{C+\sqrt{C^2+12AE}} - 2C} \nonumber \\
e^{\sqrt{\frac{C+\sqrt{C^2+12AE}}{2E}}
\left [ E + \frac{4AE^2}{(C+\sqrt{C^2+12AE})^2} + \frac{2CE}{C+\sqrt{C^2+12AE}}
\right ] } .\nonumber
\end{eqnarray}

\section{Beyond the saddle point approximation}
Saddle point methods are often used to approximate integrals that have no closed
form expression. However, it is worth noting that if we allow our result to be a
Taylor series, the inverse Laplace transforms for the partition functions we have
found \textit{can} be evaluated exactly. In its most general form, before rotating
the axis, the inverse Laplace transform is
\begin{equation}
\mathcal{L}^{-1}[Z](E) = \frac{1}{2 \pi i} \lim_{T \rightarrow \infty} 
\int_{\epsilon - iT}^{\epsilon + iT} e^{\beta E} Z(\beta)
\textup{d}\beta \nonumber .\\
\end{equation}
For the moment we will consider the extensive partition function
$Z = e^{A / \beta^d}$. The following manipulations rely on a few straightforward
identities involving hyperbolic functions.
\begin{eqnarray}
\rho(E) &=& \frac{1}{2 \pi i} \lim_{T \rightarrow \infty} 
\int_{\epsilon - iT}^{\epsilon + iT} e^{\beta E} e^{\frac{A}{\beta^d}} 
\textup{d}\beta \nonumber \\
&=& \frac{1}{2 \pi i} \lim_{T \rightarrow \infty} 
\int_{\epsilon - iT}^{\epsilon + iT} e^{\beta E} 
\left [ 1 + \frac{2}{\coth\left( \frac{A}{2\beta^d} \right) - 1} \right ] 
\textup{d}\beta \nonumber \\
&=& \frac{1}{2\pi} \int_{-\infty}^{\infty} e^{i \beta E} \textup{d}\beta + 
\frac{1}{\pi i} \lim_{T \rightarrow \infty} 
\int_{\epsilon - iT}^{\epsilon + iT} e^{\beta E} 
\frac{1}{\coth\left( \frac{A}{2\beta^d} \right) - 1} \textup{d}\beta 
\nonumber \\
&=& \delta(E) + \frac{1}{\pi i} \lim_{T \rightarrow \infty} 
\int_{\epsilon - iT}^{\epsilon + iT} e^{\beta E} 
\frac{\tanh\left(\frac{A}{2\beta^d}\right)}
{1 - \tanh\left(\frac{A}{2\beta^d}\right)} \textup{d}\beta \nonumber \\
&=& \delta(E) + \frac{1}{\pi i} \lim_{T \rightarrow \infty} 
\int_{\epsilon - iT}^{\epsilon + iT} e^{\beta E} 
\frac{\tanh\left(\frac{A}{2\beta^d}\right) + \tanh^2
\left(\frac{A}{2\beta^d}\right)}{1 - \tanh^2\left(\frac{A}{2\beta^d}\right)} 
\textup{d}\beta \nonumber \\
&=& \delta(E) + \frac{1}{\pi i} \lim_{T \rightarrow \infty} 
\int_{\epsilon - iT}^{\epsilon + iT} e^{\beta E} 
\left [ \frac{1}{2} \sinh\left(\frac{A}{2\beta^d}\right) + 
\sinh^2\left(\frac{A}{2\beta^d}\right) \right ] \textup{d}\beta \nonumber
\end{eqnarray}
To proceed further, we will Taylor expand $e^{\beta E}$. The interesting
part of the density of states (the part
without the delta function) now splits into two pieces:
\begin{equation}
\sum_{k = 0}^{\infty} \frac{E^k}{k!} \lim_{T \rightarrow \infty} 
\left [ \frac{1}{2 \pi i} \int_{\epsilon - iT}^{\epsilon + iT} 
\beta^k \sinh \left(\frac{A}{\beta^d}\right) \textup{d}\beta + 
\frac{1}{\pi i} \int_{\epsilon - iT}^{\epsilon + iT} \beta^k 
\sinh^2 \left(\frac{A}{2\beta^d}\right) \textup{d}\beta \right ] 
.\label{two-integrals}
\end{equation}
To evaluate the first integral, we will make the substitution
$\alpha = \frac{1}{\beta}$:
\begin{equation}
\frac{1}{2 \pi i} \lim_{T \rightarrow \infty} 
\int_{\frac{\epsilon - iT}{\epsilon^2 + T^2}}^
{\frac{\epsilon + iT}{\epsilon^2 + T^2}} 
\frac{\sinh \left ( A \alpha^d \right )}{\alpha^{k + 2}} 
\textup{d}\alpha .\nonumber
\end{equation}
The integrand has a pole at $\alpha = 0$ and the contour over which we are
integrating is a vertical line, slightly to the right of this pole, with a
vanishingly small height.
\begin{figure}[H]
\centering
\subfloat[][Closed contour]
{\includegraphics[width=0.45\textwidth]{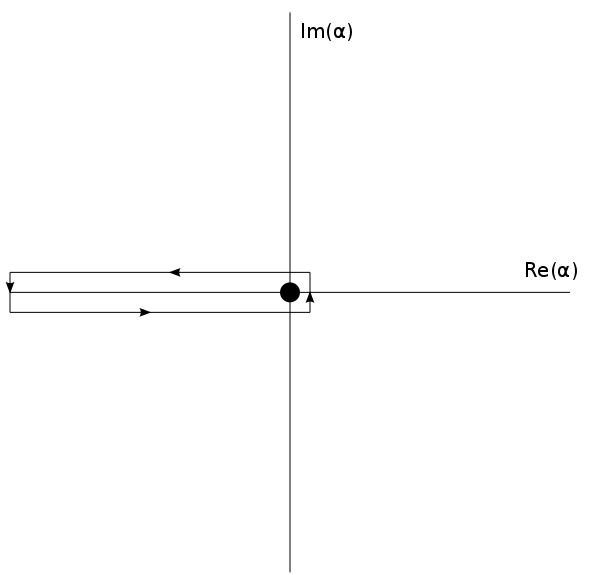}}
\subfloat[][Open contour]
{\includegraphics[width=0.45\textwidth]{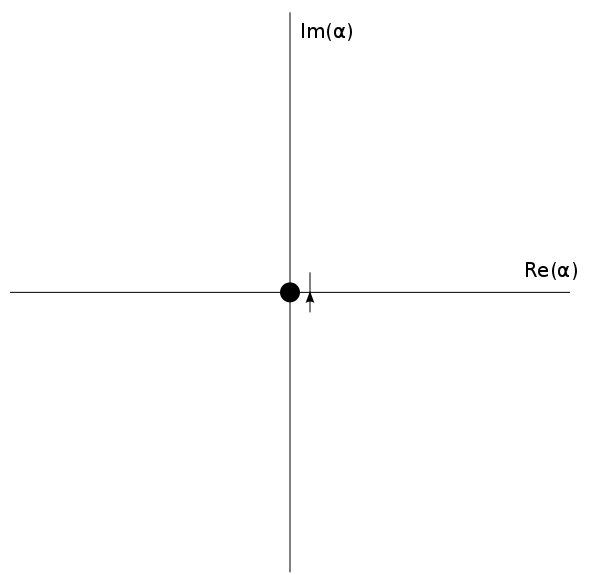}}
\caption{The integration to consider in the complex plane.}
\end{figure}
The figure on the right shows the desired contour, but fortunately,
the integral over the contour shown on the left is equal. The leftward
and rightward pieces of the contour cancel out while the downward piece
is at infinity and does not contribute. By the residue theorem,
\begin{eqnarray}
\frac{1}{2 \pi i} \lim_{T \rightarrow \infty} 
\int_{\frac{\epsilon - iT}{\epsilon^2 + T^2}}^{\frac{\epsilon + iT}
{\epsilon^2 + T^2}} \frac{\sinh \left ( A \alpha^d \right )}
{\alpha^{k + 2}} \textup{d}\alpha &=& \textup{Res} 
\left( \frac{\sinh \left ( A \alpha^d \right )}{\alpha^{k + 2}} ; 0 \right ) 
\nonumber \\
&=& \frac{1}{(k - d + 1)!} \frac{\textup{d}^{k - d + 1}}
{\textup{d}\alpha^{k - d + 1}} \left. \frac{\sinh \left ( A \alpha^d \right )}
{\alpha^d} \right|_{\alpha = 0} \nonumber \\
&=& \frac{1}{(k - d + 1)!} \frac{\textup{d}^{k - d + 1}}
{\textup{d}\alpha^{k - d + 1}} \sum_{n = 0}^{\infty} 
\left. \frac{A^{2n + 1} \alpha^{2dn}}{(2n + 1)!} \right|_{\alpha = 0} 
\nonumber \\
&=& \begin{cases} \frac{A^{\frac{k + 1}{d}}}{\left ( \frac{k + 1}{d} \right )!}
 & \frac{k + 1}{d} \; \mathrm{odd} \\ 0 & \mathrm{otherwise} \end{cases} 
 .\nonumber
\end{eqnarray}
The second integral in (\ref{two-integrals}) can be computed using the same
substitution and the same contour. All that changes is the parity of the
residue:
\begin{equation}
\frac{1}{\pi i} \lim_{T \rightarrow \infty} \int_{\frac{\epsilon - iT}
{\epsilon^2 + T^2}}^{\frac{\epsilon + iT}{\epsilon^2 + T^2}} 
\frac{\sinh^2 \left ( \frac{A}{2} \alpha^d \right )}{\alpha^{k + 2}} 
\textup{d}\alpha = \begin{cases} \frac{A^{\frac{k + 1}{d}}}
{\left ( \frac{k + 1}{d} \right )!} & \frac{k + 1}{d} \; \mathrm{even} \\ 
0 & \mathrm{otherwise} \end{cases} .\nonumber
\end{equation}
We can now substitute both of these back into (\ref{two-integrals}):
\begin{eqnarray}
\rho(E) &=& \delta(E) + \sum_{\frac{k + 1}{d} \; \mathrm{odd}} 
\frac{E^k}{k!} \frac{A^{\frac{k + 1}{d}}}{\left ( \frac{k + 1}{d} \right )!} 
+ \sum_{\frac{k + 1}{d} \; \mathrm{even}} \frac{E^k}{k!} 
\frac{A^{\frac{k + 1}{d}}}{\left ( \frac{k + 1}{d} \right )!} \nonumber \\
&=& \delta(E) + \sum_{\frac{k + 1}{d} = 1}^{\infty} \frac{E^k}{k!} 
\frac{A^{\frac{k + 1}{d}}}{\left ( \frac{k + 1}{d} \right )!} \nonumber \\
&=& \delta(E) + \sum_{j = 1}^{\infty} \frac{A^j E^{dj - 1}}{j! (dj - 1)!} 
.\label{taylor-series}
\end{eqnarray}
It is interesting to note that when $d = 1$, this is the Taylor series for a
modified Bessel function. This makes (\ref{taylor-series}) reproduce the
density of states found by Loran, Sheikh-Jabbari and Vincon \cite{loran}.
For other values of $d$ the function represented by (\ref{taylor-series}) is not
as ubiquitous.

Even though the Taylor series includes infinitely many orders beyond the
saddle point, the more useful formula for a density of states is the asymptotic
series. The saddle point method yields the first two terms of the asymptotic
series for $\log \rho(E)$. There are perhaps a number of ways to derive additional
terms but we will do so making direct use of the Taylor series.

First, define the function $f(x) = \sum_{m = 1}^{\infty} b_m x^m$ where
$b_m = \frac{1}{m! (dm - 1)!}$. An asymptotic series for $f(x)$ is as good as
the asymptotic series for $\rho(E)$ because $\rho(E) = \frac{1}{E} f(AE^d)$.
We wish to find a differential equation satisfied by $f$.

The recurrence relation satisfied by the $b_m$ is 
$m (dm - 1)(dm - 2) \dots (dm - d)b_m = b_{m - 1}$.
Multiplying both sides by $x^m$ and summing,
\begin{equation}
\sum_{m = 1}^{\infty} \frac{1}{d}(md) \dots 
(md -d)b_m x^m - b_m x^{m + 1} = 0 .\nonumber
\end{equation}
If $x^m$ were $x^{dm}$, we could differentiate it $d + 1$ times to turn
it into $(md) \dots (md -d) x^{dm - d - 1}$ and then multiply by $x^{d + 1}$ 
to get it into the form above. This suggests that our differential equation is
\begin{eqnarray}
\frac{1}{d} y^{d + 1} \frac{\textup{d}^{d + 1}}{\textup{d}y^{d + 1}} 
\left. f(y^d) \right|_{y = x^{\frac{1}{d}}} - xf(x) = 0 \nonumber \\
\frac{1}{d} x^{\frac{d}{d + 1}} \frac{\textup{d}^{d + 1}}
{\left(\textup{d} x^{\frac{1}{d}}\right)^{d + 1}} f(x) - xf(x) = 0 
.\label{diffeq}
\end{eqnarray}
For a fixed $d$, this equation can be expressed in a more explicit form. 
We write the first three below.
\begin{align}
\begin{tabular}{c|l}
$d$ & \\
\hline
$1$ & $x^2 \frac{\textup{d}^2}{\textup{d}x^2} f(x) - xf(x) = 0$ \\
$2$ & $4x^3 \frac{\textup{d}^3}{\textup{d}x^3} f(x) +6x^2 
\frac{\textup{d}^2}{\textup{d}x^2} f(x) - xf(x) = 0$ \\
$3$ & $27x^4 \frac{\textup{d}^4}{\textup{d}x^4} f(x) + 108x^3 
\frac{\textup{d}^3}{\textup{d}x^3} f(x) + 60x^2 
\frac{\textup{d}^2}{\textup{d}x^2} f(x) - xf(x) = 0$
\end{tabular}
\end{align}
Solutions to an ODE can be analyzed using the method of dominant balance.
This method, which is powerful enough to find arbitrarily many terms in the
asymptotic series, is used in the appendix to find three terms.
The end result is
\begin{equation}
f(x) \sim C_d x^{\frac{1}{2(d + 1)}} 
\exp \left ( \frac{d+1}{d^{\frac{d}{d+1}}} x^{\frac{1}{d + 1}} -
\frac{(d+2)(2d+1)}{24(d+1)} d^{-\frac{1}{d+1}} x^{-\frac{1}{d+1}} \right )
.\label{f-series}
\end{equation}
This becomes
\begin{equation}
\rho(E) \sim C_d A^{\frac{1}{2(d+1)}} E^{\frac{-d-2}{2(d+1)}}
\exp \left ( \frac{d+1}{d^{\frac{d}{d+1}}} A^{\frac{1}{d+1}} E^{\frac{d}{d+1}}
- \frac{(d+2)(2d+1)}{24(d+1)} d^{-\frac{1}{d+1}} A^{-\frac{1}{d+1}}
E^{-\frac{d}{d+1}} \right ) \label{rho-series}
\end{equation}
once we plug in $x = AE^d$ and divide by $E$.
Here, $C_d$ is a constant depending only on the dimension. This must be present
because the method of dominant balance in the appendix does not specifically
solve for the asymptotic series of $f$. It solves for the asymptotic series of
anything that solves (\ref{diffeq}) which is linear. The easiest way to specify
$C_d$ is to simply demand that the first two terms of (\ref{rho-series}) agree
with the saddle point approximation. In this case,
$C_d = \frac{1}{\sqrt{2\pi}} \sqrt{\frac{d^{\frac{1}{d+1}}}{d+1}}$ and we have
derived the last of our main results, (\ref{main-result4}).

Our result (\ref{rho-series}) finds a correction term beyond the saddle point
approximation. However, to derive it, we used the continuum approximation as our
starting point. We assumed that the free energy was extensive or that the
corrections to Weyl's formula were negligible. Is the correction in the
exponent that we have found
($A^{-\frac{1}{d+1}} E^{-\frac{d}{d+1}} \propto (|M|E^d)^{-\frac{1}{d+1}}$)
consistent with such an approximation? This depends on whether the manifold has
a boundary.

The extensive free energy $F_{\textup{E}}$ is proportional to $|M|T^{d+1}$.
On a manifold with a boundary, the subextensive free energy $F_{\textup{C}}$
is proportional to $|\partial M|T^d$. We want
\begin{equation}
|F_{\textup{E}}| \gg |F_{\textup{C}}| \label{condition1}
\end{equation}
to be satisfied. The reliability of our correction term requires that
\begin{equation}
|M|E^d \not \gg 1 \label{condition2}
\end{equation}
is satisfied as well. Condition (\ref{condition1}) states that
$T \gg |\partial M| / |M|$. Since $E \propto |M|T^{d+1}$, this means
$E \gg |\partial M|^{d+1} / |M|^d$. Consistency with condition
(\ref{condition2}) states that
$|M|^{-\frac{1}{d}} \gg |\partial M|^{d+1} / |M|^d$ or
$|M|^{\frac{d-1}{d}} \gg |\partial M|$ which is satisfied for some curved shapes.
Things are different when we consider the subextensive free energy on a sphere.
In this case condition (\ref{condition1}) gives $Er \gg 1$ while condition
(\ref{condition2}) gives $Er \not \gg 1$.

To derive a consistent formula beyond the saddle point for the density of states
on a sphere, one must drop the saddle point approximation and the continuum
approximation at the same time. In other words, one must find an accurate
expression for the inverse Laplace transform,
\begin{equation}
\mathcal{L}^{-1}\left [ e^{A\beta^{-d}+C\beta^{-(d-2)}} \right ] (E) =
\left ( \mathcal{L}^{-1}\left [ e^{A\beta^{-d}} \right ] \star
\mathcal{L}^{-1}\left [ e^{C\beta^{-(d-2)}} \right ] \right ) (E) ,\nonumber
\end{equation}
where the $\star$ denotes convolution. We already know what the individual inverse
Laplace transforms are, so we need to take the convolution of two Taylor
series.\footnote{The convolution formula requires that both functions are known at
small and large arguments so taking the convolution of two asymptotic expressions
is not sufficient.} This can be done using the identity
$\int_0^E \xi^r (E-\xi)^s \textup{d}\xi = \frac{r!s!}{(r+s+1)!} E^{r+s+1}$ from
which
\begin{equation}
\rho(E) = \delta(E) + \sum_{j = 1}^{\infty} \frac{A^j E^{dj - 1}}{j! (dj - 1)!} +
\sum_{j = 1}^{\infty} \frac{C^j E^{(d-2)j - 1}}{j! ((d-2)j - 1)!} +
\sum_{j=1}^{\infty} \sum_{k=1}^{\infty} \frac{A^j C^k}{j!k!(dj + (d-2)k-1)!}
E^{dj + (d-2)k - 1} \label{double-taylor-series}
\end{equation}
follows. Unfortunately, the double sum in (\ref{double-taylor-series})
means that the procedure in the appendix cannot be repeated for it and finding the
corresponding asymptotic series is more difficult.

\begin{acknowledgments}
This work was supported in part by the Natural Sciences and Engineering Research
Council of Canada. It also benefited from helpful comments given by Mark Van
Raamsdonk, Michael McDermott and Jared Stang.
\end{acknowledgments}

\appendix
\section{The method of dominant balance}
Our goal here is to show how the asymptotic (\ref{f-series}) follows from the
differential equation (\ref{diffeq}):
\begin{eqnarray}
\frac{1}{d} x^{\frac{d}{d + 1}} \frac{\textup{d}^{d + 1}}
{\left(\textup{d} x^{\frac{1}{d}}\right)^{d + 1}} f(x) - xf(x) = 0 \nonumber \\
x^2 \frac{\textup{d}^2}{\textup{d}x^2} f(x) - xf(x) = 0 \nonumber \\
4x^3 \frac{\textup{d}^3}{\textup{d}x^3} f(x) +6x^2 
\frac{\textup{d}^2}{\textup{d}x^2} f(x) - xf(x) = 0 \nonumber \\
27x^4 \frac{\textup{d}^4}{\textup{d}x^4} f(x) + 108x^3 
\frac{\textup{d}^3}{\textup{d}x^3} f(x) + 60x^2 
\frac{\textup{d}^2}{\textup{d}x^2} f(x) - xf(x) = 0 \nonumber \\
\dots \nonumber
\end{eqnarray}
It can readily be seen that when a power of $x$ appears beside a
derivative, the exponent is equal to the order of the derivative.
\begin{fact}
The coefficient on the highest order term in (\ref{diffeq}) is $d^d$.
\end{fact}
\begin{proof}
This is clear because $\frac{\textup{d}^{d + 1}}{\textup{d}y^{d + 1}} 
f(y^d) = \frac{\textup{d}^d}{\textup{d}y^d} \left [ d y^{d - 1} 
f^{\prime}(y^d) \right ]$. Every time one differentiates $f$, one gets another
factor of $dy^{d - 1}$. Doing this $d$ more times yields
$d^{d + 1} y^{d^2 - 1} f^{(d + 1)}(y^d)$. Letting $y = x^{\frac{1}{d}}$
and multiplying by $\frac{1}{d} x^{\frac{d + 1}{d}}$,
we get $d^d x^{d + 1} f^{(d + 1)}(x)$.
\end{proof}
We will now set $f(x) = e^{S_0(x)}$ and observe that the highest order
derivative gives us the term
$d^d x^{d + 1} \left ( \frac{\textup{d}S_0}{\textup{d}x} \right )^{d + 1}$
after we cancel $e^{S_0(x)}$. The other terms will all be of the form
$\left ( \frac{\textup{d}S_0}{\textup{d}x} \right )^{n_1} 
\dots \left ( \frac{\textup{d}^jS_0}{\textup{d}x^j} \right )^{n_j} x^j$
where $n_1 + 2n_2 + \dots + jn_j = j$ and $j \leq d + 1$.
\begin{fact}
Choosing $S_0(x) \propto x^{\frac{1}{d + 1}}$ causes all such
$\left ( \frac{\textup{d}S_0}{\textup{d}x} \right )^{n_1} \dots 
\left ( \frac{\textup{d}^jS_0}{\textup{d}x^j} \right )^{n_j} x^j$
to be $o \left( x^{d + 1} \left ( \frac{\textup{d}S_0}{\textup{d}x} 
\right )^{d + 1} \right )$.
\end{fact}
\begin{proof}
Substituting our guess for $S_0(x)$ into the term above, we find that 
it is proportional to
$\left ( \frac{1}{x^{\frac{d}{d + 1}}} \right )^{n_1} \dots 
\left ( \frac{1}{x^{\frac{d}{d + 1} + j - 1}} \right )^{n_j} x^j = 
x^{j - n_1 \left( \frac{d}{d + 1} \right) - n_2 
\left( \frac{d}{d + 1} + 1 \right) - \dots - n_j 
\left( \frac{d}{d + 1} + j - 1\right)}$.
We need to see if this is $o(x)$. Clearly it is because to maximize
$j - n_1 \left( \frac{d}{d + 1} \right) - n_2 
\left( \frac{d}{d + 1} + 1 \right) - \dots - 
n_j \left( \frac{d}{d + 1} + j - 1\right)$,
we only want to be subtracting the $n_1$ part, so we let all other $n_i = 0$.
We now want to minimize $j \left ( 1 - \frac{d}{d+1} \right )$ and the way to
do this is to let $j = d + 1$, giving us $x$. Anything less
will give us something $o(x)$.
\end{proof}
With the knowledge that we can make all terms but
$d^d x^{d + 1} \left ( \frac{\textup{d}S_0}{\textup{d}x} \right )^{d + 1}$
negligible, our differential equation becomes:
\begin{eqnarray}
d^d x^{d + 1} \left ( \frac{\textup{d}S_0}{\textup{d}x} \right )^{d + 1} = x 
\nonumber \\
\frac{\textup{d}S_0}{\textup{d}x} = \frac{1}{d^{\frac{d}{d + 1}}} 
x^{\frac{1}{d + 1} - 1} \nonumber \\
S_0(x) = \frac{d + 1}{d^{\frac{d}{d + 1}}} x^{\frac{1}{d + 1}} 
.\label{s0-solution}
\end{eqnarray}
The solution to the differential equation when we drop all but one term, is
precisely the solution that we have shown in Fact 2 to be consistent with the
dropping of the terms in the first place. Thus, to zeroth order, the
asymptotic behaviour of $f(x)$ must be
$\log{f(x)} \sim \left(\frac{(d+1)^{d + 1}}{d^d}\right)^{\frac{1}{d + 1}} 
x^{\frac{1}{d + 1}}$.
With a little more work, we can figure out what it is to first order.
\begin{fact}
The coefficient on the second highest order derivative in (\ref{diffeq})
is $d^d \frac{d^2 - 1}{2}$.
\end{fact}
\begin{proof}
Again, starting with
$\frac{\textup{d}^{d + 1}}{\textup{d}y^{d + 1}} f(y^d) = 
\frac{\textup{d}^d}{\textup{d}y^d} \left [ d y^{d - 1} 
f^{\prime}(y^d) \right ]$,
we have to differentiate $y^{d - 1}$ once and differentiate $f$ every other
time in order to get an $f^{(d)}(y^d)$. If we wait until we have
$\frac{\textup{d}}{\textup{d}y} \left [ d^d y^{d(d - 1)} 
f^{(d)}(y^d) \right ]$, the last derivative gives us
$d^d d (d - 1) y^{d(d - 1) - 1} f^{(d)}(y^d)$. If we were to differentiate
the factor of $y$ earlier when it had $y^{r(d - 1)}$, we would get
$d^d r (d - 1) y^{d(d - 1) - 1} f^{(d)}(y^d)$. Adding these up, we get
$d^d (d -1) \left ( \sum_{s = 1}^d s \right ) y^{d^2 - d - 1} 
f^{(d)}(y^d) = d^{d + 1} \frac{d^2 - 1}{2} y^{d^2 - d - 1} f^{(d)}(y^d)$.
Using $y = x^{\frac{1}{d}}$ and postmultiplying by
$\frac{1}{d} x^{\frac{d + 1}{d}}$ again, this becomes
$d^d \frac{d^2 - 1}{2} x^d f^{(d)}(x)$.
\end{proof}
Since we have a solution for $S_0(x)$, we will now let
$f(x) = e^{S_0(x) + S_1(x)}$. When this is inserted into
(\ref{diffeq}), the highest power of $x$ multiplying a single derivative of
$S_1(x)$ will come from the first term in the expansion of
\begin{equation}
d^d x^{d + 1} \left ( \frac{1}{d + 1} \frac{d + 1}{d^{\frac{d}{d + 1}}} 
x^{-\frac{d}{d + 1}} + \frac{\textup{d}S_1}{\textup{d}x} \right )^{d + 1} 
\nonumber
\end{equation}
and it will be $d \left ( \frac{d}{d + 1} \right )^{d - 1} 
\left ( \frac{(d + 1)^{d + 1}}{d^d} \right )^{\frac{d}{d + 1}} 
x^{d + 1 - \frac{d^2}{d + 1}} \left ( \frac{\textup{d}S_1}{\textup{d}x} 
\right )$.
Also, by looking at the zeroth term in the expansion, we see that an $x$ exactly
cancels the $-x$ present in the definition of (\ref{diffeq}).
\begin{fact}
Choosing $S_1(x) \propto \log{x}$ makes
$x^{d + 1 - \frac{d^2}{d + 1}} 
\left ( \frac{\textup{d}S_1}{\textup{d}x} \right ) = x^{\frac{d}{d + 1}}$
the dominant term.
\end{fact}
\begin{proof}
Recall that when $S_0(x)$ was the exponent, we had terms that looked like
$\left ( \frac{\textup{d}S_0}{\textup{d}x} \right )^{n_1} 
\dots \left ( \frac{\textup{d}^jS_0}{\textup{d}x^j} \right )^{n_j} x^j$
where $n_1 + 2n_2 + \dots + jn_j = j$ and $j \leq d + 1$. This time, we replace
$S_0(x)$ with $x^{\frac{1}{d + 1}} + S_1(x)$ so this product looks like
$\left ( x^{\frac{1}{d + 1} - 1} + 
\frac{\textup{d}S_1}{\textup{d}x} \right )^{n_1} \dots 
\left ( x^{\frac{1}{d + 1} - j} + \frac{\textup{d}^jS_1}{\textup{d}x^j} 
\right )^{n_j} x^j$.
Terms in the expansion of this product look like:
\begin{eqnarray}
&& x^{k_1 \left ( \frac{1}{d + 1} - 1\right )} 
\left ( \frac{\textup{d}S_1}{\textup{d}x} \right )^{n_1 - k_1} 
\dots x^{k_j \left ( \frac{1}{d + 1} - j\right )} 
\left ( \frac{\textup{d}^jS_1}{\textup{d}x^j} \right )^{n_j - k_j} x^j 
\nonumber \\
&\propto& x^j x^{\frac{k_1 + \dots + k_j}{d + 1} - k_1 - 2k_2 - 
\dots - jk_j - (n_1 - k_1) - 2(n_2 - k_2) - \dots -j(n_j - k_j)} \nonumber \\
&=& x^{\frac{k_1 + \dots +k_j}{d + 1}} .\nonumber
\end{eqnarray}
To maximize the exponent, we must turn it into
$\frac{n_1 + \dots + n_j}{d + 1}$ and the way to maximize this is to
let $n_1 = j = d + 1$. However, we have aleady mentioned that the $x$
term should vanish so the next highest power of $x$ we can get is
$x^{\frac{d}{d + 1}}$. We realize this either with $n_1 = j = d$ or with
$j = d + 1$ and $k_1 = n_1 - 1$. Everything else is
$o \left ( x^{\frac{d}{d + 1}} \right )$.
\end{proof}
This allows us to drop most of the terms in (\ref{diffeq}) just like before.
Specifically, we will drop all powers of $x$, lower than $x^{\frac{d}{d + 1}}$
and find that this is consistent with having $S_1(x) \propto \log{x}$. The next
two facts are devoted to solving for the coefficients that appear beside
$x^{\frac{d}{d + 1}}$.
\begin{fact}
$\frac{\textup{d}^n}{\textup{d}x^n} e^{rx^m}$ can be written as
$x^{-n} e^{rx^m} \sum_{j = 0}^n a_j (rx^m)^j$.
\end{fact}
\begin{proof}
We are dealing with
$\frac{\textup{d}^{n - 1}}{\textup{d}x^{n - 1}} 
\left ( mrx^{m - 1} e^{rx^m} \right )$. If we
differentiate the exponent $k$ times, we get $m^k r^k x^{k(m - 1)} e^{rx^m}$.
If we then differentiate the power of $x$ $n - k$ times, our result is
proportional to the exponential times
$r^k x^{k(m - 1) - (n - k)} = x^{-n} (rx^m)^k$.
\end{proof}
In the situations where we need to apply this fact, $m = \frac{1}{d + 1}$ and
$r = \frac{d + 1}{d^{\frac{d}{d + 1}}}$. One way to get an
$x^{\frac{d}{d + 1}}$ term is to let $n = j = d$. In this case, the coefficient
we find is
$r^d m^d = \left ( \frac{(d + 1)^{d + 1}}{d^d} \right )^{\frac{d}{d + 1}} 
\frac{1}{(d + 1)^d}$. We must remember that this is multiplied by the
coefficient on the $d^{\mathrm{th}}$ derivative found in Fact 3.
\begin{fact}
The other route to an $x^{\frac{d}{d + 1}}$ term - choosing $n = d + 1$ and
$j = d$ - yields a coefficient of
$\frac{-d^2}{2} \frac{1}{d^{\frac{d^2}{d + 1}}}$.
\end{fact}
\begin{proof}
We have
$\frac{\textup{d}^d}{\textup{d}x^d} \left ( mr x^{m - 1} e^{rx^m} \right )$
and we have to differentiate the exponential all but one time. Waiting until
$\frac{\textup{d}}{\textup{d}x} \left ( m^d r^d x^{d(m - 1)} e^{rx^m} \right)$,
this differentiates to $d (m - 1) m^d r^d x^{dm - d - 1} e^{rx^m}$. Had we done
this for $x^{s(m - 1)}$, we would have had $s (m - 1) m^d r^d x^{dm - d - 1}$
times the exponential. We can therefore add these up to get
$\frac{-d}{d + 1} \left ( \frac{1}{d + 1} \right )^d 
\left ( \frac{(d + 1)^{d + 1}}{d^d} \right )^{\frac{d}{d + 1}} 
\sum_{s = 1}^d s = \frac{-d^2}{2} \frac{1}{d^{\frac{d^2}{d + 1}}}$.
\end{proof}
We must also remember that this is multiplied by the coefficient on the
$(d + 1)^{\mathrm{st}}$ derivative found in Fact 1.
The following equation is the result of keeping only the terms found to be
significant in Fact 4 complete with prefactors. The left hand side uses the
highest power of $x$ multiplying a single derivative of $S_1(x)$ and the right
hand side uses the two coefficients found above corresponding to the two ways
of constructing $x^{\frac{d}{d + 1}}$.
\begin{eqnarray}
d \left ( \frac{d}{d + 1} \right )^{d - 1} \left ( \frac{(d + 1)^{d + 1}}{d^d} 
\right )^{\frac{d}{d + 1}} x^{1 + \frac{d}{d + 1}} \left ( 
\frac{\textup{d}S_1}{\textup{d}x} \right ) &=& \left [ d^d \frac{d^2}{2} 
\frac{1}{d^{\frac{d^2}{d + 1}}} - d^d \frac{d^2 - 1}{2} 
\frac{1}{d^{\frac{d^2}{d + 1}}} \right ] x^{\frac{d}{d + 1}} \nonumber \\
\frac{\textup{d}S_1}{\textup{d}x} &=& \frac{1}{x} \left [ \frac{d^2}{2(d + 1)}
 - \frac{d - 1}{2} \right ] \nonumber \\
	S_1(x) &=& \frac{1}{2(d + 1)} \log{x} \label{s1-solution}
\end{eqnarray}
If we stop here, we will have completed no more than a roundabout derivation
of the saddle point result. We will set $f(x) = e^{S_0(x) + S_1(x) + S_2(x)}$
and solve for $S_2(x)$.
\begin{fact}
The third coefficient in (\ref{diffeq}) is
$\frac{1}{24}d^{d-1}(d-2)(d-1)^2(d+1)(3d+1)$.
\end{fact}
\begin{proof}
$\frac{\textup{d}^{d + 1}}{\textup{d}y^{d + 1}} f(y^d) = 
\frac{\textup{d}^d}{\textup{d}y^d} \left [ d y^{d - 1} 
f^{\prime}(y^d) \right ]$ We have to differentiate $f$ all but two times.
If $y$ has the exponent $s(d-1)$ when we differentiate it the first time, we will
bring down a coefficient of $s(d-1)$. The second time, $y$ can have an exponent
of $t(d-1)-1$ where $s \leq t \leq d-1$. Therefore, we get
$d^{d-1} \sum_{s=1}^{d-1} s(d-1) \sum_{t=s}^{d-1}[t(d-1)-1] y^{(d-1)^2-2}
f^{(d-1)}(y^d) = \frac{1}{24}d^d(d-2)(d-1)^2(d+1)(3d+1) y^{(d-1)^2-2}
f^{(d-1)}(y^d)$. We must substitute $y = x^{\frac{1}{d}}$ and postmultiply by
$\frac{1}{d} x^{\frac{d + 1}{d}}$ to get
$\frac{1}{24}d^{d-1}(d-2)(d-1)^2(d+1)(3d+1) x^{d-1} f^{(d-1)}(x)$.
\end{proof}
Just like before, we want to look for the highest power of $x$ multiplying
a single derivative of $S_2(x)$. This comes from the expansion of
\begin{equation}
d^dx^{d+1} \left ( \frac{1}{d + 1} \frac{d + 1}{d^{\frac{d}{d + 1}}} 
x^{-\frac{d}{d + 1}} + \frac{1}{2(d+1)x} + \frac{\textup{d}S_2}{\textup{d}x}
\right )^{d+1} \nonumber
\end{equation}
where the first two terms in the brackets are the derivatives of $S_0(x)$ and
$S_1(x)$. The term that will become important is
$(d+1)d^{\frac{d}{d+1}} x^{\frac{2d+1}{d+1}} \frac{\textup{d}S_2}{\textup{d}x}$.
\begin{fact}
$\frac{\textup{d}^n}{\textup{d}x^n} \left ( x^p e^{rx^m} \right )$ can be written
as $x^{-n} x^p e^{rx^m} \sum_{q=0}^n \binom{n}{q}p\dots(p-n+q+1)
\sum_{j=0}^q a_j (rx^m)^j$.
\end{fact}
\begin{proof}
We can use the product rule to write our derivative as
$\sum_{q=0}^n \binom{q}{n} \frac{\textup{d}^q}{\textup{d}x^q} e^{rx^m}
\frac{\textup{d}^{n-q}}{\textup{d}x^{n-q}} x^p = \sum_{q=0}^n \binom{q}{n}
\frac{\textup{d}^q}{\textup{d}x^q} e^{rx^m} p\dots(p-n+q+1) x^{p-n+q}$ and then
apply Fact 5.
\end{proof}
Again, we plan on using this relation when $m = \frac{1}{d+1}$,
$r = \frac{d+1}{d^{\frac{d}{d+1}}}$ and $p = \frac{1}{2(d+1)}$.
When we were solving for $S_1(x)$ we saw that the highest powers of $x$ appearing
in the differential equation ($x$ itself) cancelled. Now that we are solving for
$S_2(x)$, we can show that the next highest power of $x$ that could potentially
appear on its own will cancel as well.
\begin{fact}
The coefficient on $x^{\frac{d}{d+1}}$ vanishes.
\end{fact}
\begin{proof}
By looking at Fact 8, we can see that we get $x^{\frac{d}{d+1}}$ when $j = d$.
This can be realized with $(n = q = j = d)$, $(n = d+1, q = j = d)$, or
$(n = q = d+1, j = d)$. Remembering the appropriate overall coefficients that
come with $n = d$ and $n = d + 1$, we can calculate the contribution due to each
case. The first contributes $d^d \frac{d^2-1}{2}(rm)^d$, the second
$d^d(d+1)p(rm)^d$ and the third $d^d(m-1)(rm)^d \sum_{s=1}^d s$. All in all,
we get $d^{\frac{d}{d+1}} \left [ \frac{d^2-1}{2}+\frac{1}{2}
-\frac{d}{d+1}\frac{d(d+1)}{2} \right ] = 0$.
\end{proof}
What this tells us is that the lone power of $x$ whose coefficient we need to
find is $x^{\frac{d-1}{d+1}}$.
\begin{fact}
Choosing $S_2(x) \propto x^{-\frac{1}{d+1}}$ makes
$x^{\frac{2d+1}{d+1}}\frac{\textup{d}S_2}{\textup{d}x}\propto x^{\frac{d-1}{d+1}}$
the dominant term.
\end{fact}
\begin{proof}
We must consider $x^j$ multiplied by powers of the $j^{\mathrm{th}}$ derivative
or lower order derivatives of $S_0(x) + S_1(x) + S_2(x)$. This leads to
$\left ( x^{\frac{1}{d+1}-1} + \frac{1}{x} + 
\frac{\textup{d}S_2}{\textup{d}x} \right )^{n_1} \dots \left ( x^{\frac{1}{d+1}-j}
+ \frac{1}{x^j} + \frac{\textup{d}^jS_2}{\textup{d}x^j} \right )^{n_j} x^j$ where
$n_1 + 2n_2 + \dots + jn_j = j \leq d + 1$. Terms in this product take the form:
\begin{eqnarray}
&& x^{k_1\left( \frac{1}{d+1} -1 \right)} \left ( \frac{1}{x} \right )^{l_1}
\left ( \frac{\textup{d}S_2}{\textup{d}x} \right )^{n_1-k_1-l_1} \dots
x^{k_j\left( \frac{1}{d+1} -j \right)} \left ( \frac{1}{x^j} \right )^{l_j}
\left ( \frac{\textup{d}^jS_2}{\textup{d}x^j} \right )^{n_j-k_j-l_j} \nonumber \\
&\propto& x^jx^{\frac{k_1+\dots+k_j}{d+1}-k_1-2k_1-\dots-jk_j-l_1-2l_2-\dots-jl_j
-(n_1-k_1-l_1)\left( \frac{1}{d+1} + 1\right) - \dots -
-(n_j-k_j-l_j)\left( \frac{1}{d+1} + j\right)} \nonumber \\
&=& x^{2\frac{k_1+\dots+k_j}{d+1} + \frac{l_1+\dots+l_j}{d+1} -
\frac{n_1+\dots+n_j}{d+1}} \nonumber
\end{eqnarray}
If we wanted to maximize the exponent, we would turn the power of $x$ into
$x^{\frac{n_1+\dots+n_j}{d+1}}$ and turn it into $x$ by letting $j = d + 1$.
However, we already know that $x$ should cancel, so the next highest power of
$x$ we could make is $x^{\frac{d}{d+1}}$ either with $j = d + 1$ or $j = d$.
This cancels too so we maximize the exponent with $x^{\frac{d-1}{d+1}}$ which
can be made in a number of ways with $j = d + 1$, $j = d$ or $j = d - 1$.
\end{proof}
Solving for the coefficient in front of $x^{\frac{d-1}{d+1}}$ is the last task
that remains. Looking at Fact 8, we can see that we need $mj = \frac{d-1}{d+1}$
or $j = d - 1$. There are six ways to get this:
\begin{enumerate}
\item $n = d - 1, q = d - 1, j = d - 1$
\item $n = d, q = d - 1, j = d - 1$
\item $n = d, q = d, j = d - 1$
\item $n = d + 1, q = d - 1, j = d - 1$
\item $n = d + 1, q = d, j = d - 1$
\item $n = d + 1, q + d + 1, j = d - 1$
\end{enumerate}
From Facts 1, 3, and 7, we know the coefficients that come from the values
of $n$. We also know that the coefficient coming from $q$ is a binomial
coefficient from Fact 8. We need to find the values of $a_j$ where
$j$ differs from $q$ either by zero, one, or two. Expanding the derivative
and using the same logic we used in the proof of Fact 7,
\begin{eqnarray}
\frac{\textup{d}^q}{\textup{d}x^q} \left ( e^{rx^m} \right ) &=&
\frac{\textup{d}^{q-1}}{\textup{d}x^{q-1}} \left ( rx^{m-1} e^{rx^m} \right )
\nonumber \\
&=& (rmx^{m-1})^q e^{rx^m} + (m-1)(rm)^{q-1} \sum_{s=1}^{q-1}s x^{(q-1)(m-1)-1}
e^{rx^m} \nonumber \\
&& + (rm)^{q-2} \sum_{s=1}^{q-2} s(m-1) \sum_{t=s}^{q-2} [t(m-1)-1]
x^{(q-2)(m-1)-1} e^{rx^m} + \dots \nonumber \\
&=& (rmx^{m-1})^q e^{rx^m} + (m-1)(rm)^{q-1} \frac{q(q-1)}{2} \nonumber \\
&& + (rm)^{q-2} [\frac{1}{24} (m-1)^2 (q-1)(q-2)(3(q-1)^2+q-3)
- \frac{1}{6}(m-1)q(q-1)(q-2)] + \dots \nonumber
\end{eqnarray}
we can read off the relevant $a_j$ values. Constructing the six coefficients,
we get
\begin{enumerate}
\item $\frac{1}{24} d^{d-1} (d-1)^2(d-2)(d+1)(3d+1)(rm)^{d-1}
= \frac{1}{24} d^{\frac{d-1}{d+1}} (d-2)(d-1)^2(d+1)(3d+1)$
\item $d^d \frac{d^2-1}{2} dp (rm)^{d-1} = d^{\frac{3d+1}{d+1}} \frac{d-1}{4}$
\item $d^d \frac{d^2-1}{2} (m-1)(rm)^{d-1} \frac{d(d-1)}{2}
= -d^{\frac{4d+2}{d+1}} \frac{(d-1)^2}{4}$
\item $d^d \frac{d(d+1)}{2} p(p-1)(rm)^{d-1} =
-\frac{1}{8} d^{\frac{3d+1}{d+1}} \frac{2d+1}{d+1}$
\item $d^d (d+1)p(m-1)(rm)^{d-1} \frac{d(d-1)}{2} =
-\frac{1}{4}d^{\frac{4d+2}{d+1}} \frac{d-1}{d+1}$
\item $\frac{1}{24} d^d (rm)^{d-1} [(m-1)^2d(d-1)(3d^2+d-2) - 4(m-1)(d-1)d(d+1)]
= \frac{1}{24} d^{\frac{4d+2}{d+1}} (d-1) \frac{3d^2+2d+4}{d+1}$
\end{enumerate}
which must be added up to give
$-\frac{1}{24} d^{\frac{d-1}{d+1}} \frac{(d+2)(2d+1)}{d+1}$ as the coefficient
appearing beside $x^{\frac{d-1}{d+1}}$.
\begin{eqnarray}
(d+1)d^{\frac{d}{d+1}} x^{\frac{2d+1}{d+1}} \frac{\textup{d}S_2}{\textup{d}x}
&=& \frac{1}{24} d^{\frac{d-1}{d+1}} \frac{(d+2)(2d+1)}{d+1} x^{\frac{d-1}{d+1}}
\nonumber \\
\frac{\textup{d}S_2}{\textup{d}x} &=& \frac{1}{24} d^{-\frac{1}{d+1}}
\frac{(d+2)(2d+1)}{(d+1)^2} x^{\frac{-d-2}{d+1}} \nonumber \\
S_2(x) &=& \frac{1}{24}d^{-\frac{1}{d+1}} \frac{(d+2)(2d+1)}{d+1}
x^{-\frac{1}{d+1}} \label{s2-solution}
\end{eqnarray}
Putting (\ref{s0-solution}), (\ref{s1-solution}), and (\ref{s2-solution}) into
$f(x) = e^{S_0(x)+S_1(x)+S_2(x)}$, we get the asymptotic form:
\begin{equation}
f(x) \sim x^{\frac{1}{2(d + 1)}} 
\exp \left ( \frac{d+1}{d^{\frac{d}{d+1}}} x^{\frac{1}{d + 1}} -
\frac{(d+2)(2d+1)}{24(d+1)} d^{-\frac{1}{d+1}} x^{-\frac{1}{d+1}} \right ) .
\end{equation}

\bibliographystyle{unsrt}
\bibliography{references}

\begin{thebibliography}{10}

\bibitem{cardy}
J.~L. Cardy.
\newblock Operator content of two-dimensional conformally invariant theories.
\newblock {\em Nuclear Physics B}, 270:186--204, 1986.

\bibitem{farey}
R.~Dijkgraaf; J. Maldacena; G. Moore;~E. Verlinde.
\newblock A black hole {F}arey-tail.
\newblock 2000.
\newblock hep-th/0005003.

\bibitem{loran}
F.~Loran; M. M. Sheikh-Jabbari;~M. Vincon.
\newblock Beyond logarithmic corrections to {C}ardy formula.
\newblock {\em Journal of High Energy Physics}, 110, 2011.
\newblock arXiv:1010.3561.

\bibitem{verlinde2}
E.~Verlinde.
\newblock On the holographic principle in a radiation dominated universe.
\newblock 2000.
\newblock hep-th/0008140.

\bibitem{klemm}
D.~Klemm; A. C. Petkou;~G. Siopsis.
\newblock Entropy bounds, monotonicity properties and scaling in {CFT}s.
\newblock {\em Nuclear Physics B}, 601:380--394, 2001.
\newblock hep-th/0101.076.

\bibitem{cai2}
R.~G. Cai.
\newblock Cardy-{V}erlinde formula and {A}d{S} black holes.
\newblock {\em Physical Review D}, 63, 2001.
\newblock hep-th/0102113.

\bibitem{lin}
F.~L. Lin.
\newblock Casimir effect of graviton and the entropy bound.
\newblock {\em Physical Review D}, 63, 2001.
\newblock hep-th/0010127.

\bibitem{brustein}
R.~Brustein; S. Foffa;~G. Veneziano.
\newblock {CFT}, holography and causal entropy bound.
\newblock {\em Physics Letters B}, 507:270--276, 2001.
\newblock hep-th/0101083.

\bibitem{gubser}
S.~S. Gubser; I. R. Klebanov; A.~A. Tseytlin.
\newblock Coupling constant dependence in the thermodynamics of {N} = 4
  supersymmetric {Y}ang-{M}ills theory.
\newblock {\em Nuclear Physics B}, 534:202--222, 1998.
\newblock hep-th/9805156.

\bibitem{kim}
C.~Kim; S.~J. Rey.
\newblock Thermodynamics of large-{N} super {Y}ang-{M}ills theory and
  {A}d{S}/{CFT} correspondence.
\newblock {\em Nuclear Physics B}, 564, 2000.
\newblock hep-th/9905205.

\bibitem{fotopoulos}
A.~Fotopoulos; T.~R. Taylor.
\newblock Remarks on two-loop free energy in {N}=4 supersymmetric
  {Y}ang-{M}ills theory at finite temperature.
\newblock {\em Physical Review D}, 59, 1999.
\newblock hep-th/9811224.

\bibitem{cai1}
R.~G. Cai;~N. Ohta.
\newblock Thermodynamics of large {N} noncommutative super {Y}ang-{M}ills
  theory.
\newblock {\em Physical Review D}, 61, 2000.
\newblock hep-th/9910092.

\bibitem{braaten}
E.~Braaten;~A. Nieto.
\newblock Free energy of {QCD} at high temperature.
\newblock {\em Physical Review D}, 53:3421--3437, 1996.

\bibitem{ivrii}
V.~Ivrii.
\newblock Second term of the spectral asymptotic expansion of the
  {L}aplace-{B}eltrami operator on manifolds with boundary.
\newblock {\em Funcitonal Analysis and its Applications}, 14(2):98--106, 1980.

\bibitem{gutzwiller}
M.~C. Gutzwiller.
\newblock {\em Chaos in Classical and Quantum Mechanics}.
\newblock Springer, 1990.

\bibitem{higuchi}
A.~Higuchi.
\newblock Symmetric tensor spherical harmonics and their application to the
  de-{S}itter group {SO}({N}, 1).
\newblock {\em Journal of Mathematical Physics}, 28, 1986.

\bibitem{kutasov}
D.~Kutasov;~F. Larsen.
\newblock Partition sums and entropy bounds in weakly coupled {CFT}.
\newblock {\em Journal of High Energy Physics}, 1, 2001.
\newblock hep-th/0009244.

\bibitem{brevik}
I.~Brevik; K. A. Milton; S.~D. Odintsov.
\newblock Entropy bounds in {R}x{S}3 geometries.
\newblock {\em Annals of Physics}, 302:120--141, 2002.
\newblock hep-th/0202048.

\end{thebibliography}
\end{document}